%% file: ffsat.tex
\PassOptionsToPackage{usenames,dvipsnames}{color}
\documentclass{easychair}
\usepackage[T1]{fontenc}
\usepackage{lmodern}
\usepackage{graphicx}

\usepackage{amsmath}    %
\usepackage{amssymb}    %
\usepackage{mathtools}  %
\usepackage{microtype}  %
\usepackage[inline]{enumitem} %
\usepackage{multirow}   %
\usepackage{booktabs}   %
\usepackage{subcaption} %

\usepackage[ruled,linesnumbered,vlined]{algorithm2e} %

\usepackage{enumitem}
\usepackage[]{xcolor}
\usepackage{nag}       %
\usepackage[babel]{csquotes}
\usepackage{xcolor}
\usepackage{listings}
\usepackage{dirtree}  
\usepackage{stmaryrd}
\usepackage{array}
\usepackage{xspace}
\xspaceaddexceptions{]\}}
\usepackage{forest}

\usepackage{enumitem}
\setlist{nosep}

\usepackage{cite} %
\usepackage{hhline}
\usepackage{multirow,changepage}
\usepackage[linewidth=1pt]{mdframed}
\usepackage{hyperref}

\usepackage[acronym,nohypertypes={acronym,notation,index}]{glossaries} %

\newacronym{mcsat}{MCSat}{Model Constructing Satisfiability}
\newacronym{dpll}{DPLL}{Davis–Putnam–Logemann–Loveland}
\newacronym{cdcl}{CDCL}{Conflict-Driven Clause Learning}
\newacronym{cnf}{CNF}{conjunctive normal form}
\newacronym{dnf}{DNF}{disjunctive normal form}
\newacronym{sat}{SAT}{boolean satisfiability}
\newacronym{smt}{SMT}{satisfiability modulo theory}
\newacronym{dpll-t}{DPLL(T)}{DPLL with theory support}
\newacronym{cad}{CAD}{cylindrical algebraic decomposition}
\newacronym{srs}{SRS}{subresultant regular subchain}
\newacronym{gcd}{gcd}{greatest common divisor}

\newcommand{\sfUndef}{\ensuremath{\mathsf{undef}}}
\newcommand{\sfTrue}{\ensuremath{\mathsf{true}}}
\newcommand{\sfFalse}{\ensuremath{\mathsf{false}}}
\newcommand{\sfValue}[2]{\ensuremath{\mathsf{val}(#1,#2)}}
\newcommand{\sfConstraints}[1]{\ensuremath{\mathsf{constr}(#1)}}
\newcommand{\sfFeasible}[1]{\ensuremath{\mathsf{fsbl}(#1)}}
\newcommand{\sfLevel}[1]{\ensuremath{\mathsf{level}(#1)}}
\newcommand{\sfCompatible}[2]{\ensuremath{\mathsf{comp}(#1,#2)}}
\newcommand{\sfExplain}[2]{\ensuremath{\mathsf{exp}(#1,#2)}}
\newcommand{\sfResolve}[3]{\ensuremath{\mathsf{resolve}(#1,#2,#3)}}
\newcommand{\sfSat}{\ensuremath{\mathsf{sat}}}
\newcommand{\sfUnsat}{\ensuremath{\mathsf{unsat}}}

\newcommand{\justifiy}{\ensuremath{\!\!\rightarrow\!\!}}

\newcommand{\true}{\sfTrue}
\newcommand{\false}{\sfFalse}

\newcommand{\Fp}{\ensuremath{\mathbb{F}_q}}

\newcommand{\FpX}{\ensuremath{\Fp[X]}}
\newcommand{\FpXk}{\ensuremath{\Fp[X_k]}}
\newcommand{\FpXkm}{\ensuremath{\Fp[X_{k-1}]}}
\newcommand{\zero}{\ensuremath{\mathsf{zero}}}
\newcommand{\zeroq}{\ensuremath{\mathsf{zero}_{q}}}
\newcommand{\proj}{\ensuremath{\mathsf{proj}_k\xspace}}
\newcommand{\FpO}{\ensuremath{\overline{\set F}_q}}

\newcommand{\lv}[1]{\ensuremath{\mathsf{lv}(#1)}}

\newcommand{\lc}[2]{\ensuremath{\mathsf{lc}(#1,#2)}}
\newcommand{\ldeg}[2]{\ensuremath{\mathsf{deg}(#1,#2)}}
\newcommand{\red}{\ensuremath{\mathsf{red}}}
\newcommand{\var}[1]{\ensuremath{\mathsf{vars}(#1)}}
\newcommand{\level}[1]{\ensuremath{\mathsf{level}(#1)}}
\newcommand{\poly}[1]{\ensuremath{\mathsf{poly}(#1)}}
\newcommand{\cls}{\ensuremath{\mathsf{cls}}}

\newcommand{\prem}{\ensuremath{\mathsf{prem}}}
\newcommand{\pquo}{\ensuremath{\mathsf{pquo}}}
\newcommand{\srs}{\ensuremath{\mathsf{srs}}}
\newcommand{\alphaB}{\ensuremath{\boldsymbol{\alpha}}}
\newcommand{\xiB}{\ensuremath{\boldsymbol{\xi}}}
\newcommand{\alphaM}{\ensuremath{\alphaB_M}}

\newcommand{\C}{\ensuremath{\mathcal{C}}}

\let\set\mathbb

\newtheorem{definition}{Definition}
\newtheorem{example}{Example}
\newtheorem{theorem}{Theorem}

\newcommand\rulename[1]{\textsc{#1}}
\newcommand{\SelectClauseRule}{\rulename{Sel-Clause}\xspace}
\newcommand{\ConflictRule}{\rulename{Conflict}\xspace}
\newcommand{\SatRule}{\rulename{Sat}\xspace}
\newcommand{\LiftLevelRule}{\rulename{Lift-Level}\xspace}
\newcommand{\BPropRule}{\rulename{B-Prop}\xspace}
\newcommand{\DecideLitRule}{\rulename{Decide-Lit}\xspace}
\newcommand{\TPropRule}{\rulename{T-Prop}\xspace}
\newcommand{\ResolvePropRule}{\rulename{Resolve-Prop}\xspace}
\newcommand{\ResolveDecRule}{\rulename{Resolve-Dec}\xspace}
\newcommand{\ConsumePropRule}{\rulename{Consume-Prop}\xspace}
\newcommand{\ConsumeDecRule}{\rulename{Consume-Dec}\xspace}
\newcommand{\DropLevelOneRule}{\rulename{Drop-Level-1}\xspace}
\newcommand{\DropLevelTwoRule}{\rulename{Drop-Level-2}\xspace}
\newcommand{\UnsatRule}{\rulename{Unsat}\xspace}

\makeatletter
\patchcmd{\@algocf@start}%
  {-1.5em}%
  {0pt}%
  {}{}%
\makeatother

\begin{document}
\title{SMT Solving over Finite Field Arithmetic}
\titlerunning{SMT Solving over Finite Field Arithmetic}
\author{Thomas Hader \and
Daniela Kaufmann \and
Laura Kov\'acs}
\authorrunning{T. Hader et al.}
\institute{TU Wien, Vienna, Austria \\% Institut für Logic and Computation, \\ 
\email{\{thomas.hader,daniela.kaufmann,laura.kovacs\}@tuwien.ac.at}}
\maketitle
\begin{abstract}

Non-linear polynomial systems over finite fields are used to model 
functional behavior of cryptosystems, with applications in system security, 
computer cryptography, and post-quantum cryptography. 
Solving polynomial systems is also one of the most difficult problems in mathematics.
In this paper, we propose an automated reasoning procedure for deciding the satisfiability of a system of non-linear equations 
over finite fields. We introduce zero
decomposition techniques to prove
that polynomial constraints over finite fields yield finite basis
explanation functions. We use these explanation functions  in
model constructing satisfiability solving, allowing us to equip a CDCL-style search
procedure with tailored theory reasoning in SMT solving over finite fields. 
We implemented our approach and provide a novel and effective reasoning prototype for non-linear arithmetic over finite fields. 

\end{abstract}

\input{01_intro}  %
\input{02_prelim} %
\input{03_mcsat}  %
\input{04_explain} %
\input{05_implementation} %
\input{06_experiments} %
\input{07_related} %
\input{08_conclusion} %

\paragraph{Acknowledgements.} 
We thank Nikolaj Bj\o rner for the fruitful discussion on this work.
We acknowledge partial support from the ERC Consolidator Grant
ARTIST 101002685, the TU Wien SecInt Doctoral College, and the FWF
SFB project SpyCoDe F8504.

\bibliographystyle{splncs04}
\bibliography{ffsat}

\end{document}

%% file: 01_intro.tex
\section{Introduction}\label{sec:intro}

Solving a system of polynomial equations is one of the hardest problems in mathematics, with emerging applications
in cryptography, software security, code optimizations, control
theory, and many other areas of computer science.
Computing  solutions to polynomial
equations is known to be decidable and algorithmically
solvable over algebraically closed fields thanks to the fundamental theory of algebra and Buchberger's algorithms for Gr\"obner basis
computation~\cite{SturmfeldsSolvingSystemofPolyEqs,Buchberger}. Yet,
when restricting the algorithmic study of solving polynomial equations
over integers, the problem becomes
undecidable~\cite{Matiyasevich93}.  

Until recently, the algorithmic study of solving polynomial
constraints, and hence automated reasoning  in polynomial arithmetic, was the
sole domain of computer algebra systems~\cite{Reduce,Mathematica,maple,sage}.
These systems  are very powerful in computing the set of all solutions of
polynomial constraints, but generally suffer from high computational
overhead, such as doubly exponential computation complexities in terms of
number of variables~\cite{DBLP:journals/jsc/DavenportH88}.

With the purpose of scaling non-linear
reasoning, especially for solving satisfiability instances of
polynomial arithmetic,
exciting new developments in \gls{sat}/\gls{smt} reasoning arose by combining a \gls{cdcl}-style search for a feasible assignment, called \gls{mcsat}, with
algebraic decompositions and projections over the solution space of
polynomial inequalities~\cite{nlsat,mcsat}. 
Unlike the classic CDCL(T) approach of \gls{smt}-solvers,
\gls{mcsat}~\cite{mcsat,mcsat_inear_int,nlsat} combines the
capabilities of a SAT solver and a theory solver into a single
procedure while keeping the search principles theory independent. 
To the best of our knowledge, SMT solving over finite fields lacks a dedicated approach for reasoning over finite fields.
Encoding the problem in existing theories (e.g.\ NIA) are inefficient~\cite{ozdemir2023satisfiability}.

\emph{In this paper we address this challenge and introduce a 
CDCL-style search procedure extended with zero decomposition
techniques for explaining and resolving (variable) conflicts while solving polynomial
constraints over finited fields.}

\paragraph{Need for Finite Fields.}
Finite fields provide natural
ground to model bounded machine arithmetic, for example when considering 
modern cryptosystems with applications in system security and
post-quantum cryptography. Existing approches build for example private and secure systems from Zero-Knowledge
Proofs~\cite{GoldwaserMicaliRackoff-SIAM89} or verify blockchain
technologies, 
such as smart contracts~\cite{Szabo2018SmartC}, with all these efforts
implementing finite field arithmetic. 
Elliptic curve cryptography~\cite{ecc} exploits polynomials over
finite fields, with further use in TLS encryption~\cite{tls},
SSH~\cite{ssh} and digital signatures~\cite{ecsda}.
Polynomial equations over finite fields are also used in coding
theory~\cite{book:MW,mou-phdthesis},  decoding error-correcting codes
of large error rates. 
In addition, solving
polynomials over finite fields has applications in finite biological models, such as modeling cycles of biological networks as continuous dynamical systems~\cite{NiuWang-MCS,mou-phdthesis}.

\paragraph{SMT Solving over Finite Fields.}
In this paper we introduce 
 an \gls{mcsat}-based decision procedure for solving
polynomial constraints over finite fields, extending thus the landscape
of \gls{smt} solving with finite field arithmetic.
We formalize \emph{SMT solving over finite fields}  as
follows (see Section~\ref{sec:prelim} for relevant 
notation).

\begin{mdframed}
Given a finite field $\mathbb{F}_q$ with order $ q = p^k $, where $ p
$ is a prime number and $ k\geq 1 $, let $F$ be a set of polynomial
constraints in $\FpX$ and $\mathcal{F}$ a formula following the logical structure:
\[\mathcal{F} \quad= \quad \bigwedge_{C\subseteq F} \bigvee_{f\in C} f  \quad=  \quad \bigwedge_{C\subseteq F} \bigvee_{f\in C} \poly{f} \vartriangleright 0 \quad \text{with}   \vartriangleright \in \{=, \neq\}. \]
\emph{SMT solving over finite fields:}  Does an assignment $ \nu : \{x_1,\dots,x_n\} \rightarrow \Fp $ exists that satisfies $\mathcal{F}$?
\end{mdframed}

\begin{example}\label{ex:motivating}
  We show an instance of the SMT solving problem over finite fields,
  by considering  the finite field  $\mathbb{F}_5$ whose elements are
  $\{0,1,2,3,4\}$.
  Note that $-1$ is $4$ in   $\mathbb{F}_5$.
  Let $\mathcal{F}$ be the formula representing the conjunction of the
  polynomial constraints $\{x_1^2-1=0, x_1x_2-x_2-1=0\}$ over  $\mathbb{F}_5[x_1,x_2]$.
  In our work we address SMT solving of $\mathcal{F}$ over  $\mathbb{F}_5[x_1,x_2]$, deriving that $\mathcal{F}$ is satisfiable using the variable
  assignment $\{x_1\mapsto 4, x_2\mapsto 2\}$.
  \end{example}

To the best of our knowledge, existing \gls{smt}-based approaches lack the necessary theory for reasoning over 
finite fields, and therefore assertions that model the behavior 
of finite fields must be included in the input problem formalization
(i.e.\ $F$).
As a workaround, one may use so-called \emph{field polynomials} ($\{x_k^q-x_k\mid 1\leq k\leq n\} $ for a ring $\FpX$) to
characterize finite fields and thus restrict the solution space of $\FpX$  to the  
 finite domain of the field  $\mathbb{F}_q$. Unfortunately, using field polynomials is
 practically inefficient, as already witnessed in our initial attempts
 from~\cite{ma, hader-smt22}: when used during variable elimination, field
 polynomials yield new polynomials as logical consequences of the
initial set $F$ of polynomials and at the same time hugely increase the
degree and size of the newly derived polynomials in the search space.

\paragraph{Our contributions.}
In this paper we do not rely on field 
polynomials but 
extend the theory-dependent rules of \gls{mcsat} to natively support finite 
fields arithmetic.
The main difficulty in \gls{mcsat}-based reasoning comes with
generating so-called \emph{explanation clauses} for resolving
conflicting variable assignments during SMT solving. 
We therefore develop a novel \emph{theory propagation} rule for
finite fields that admits propagation of theory literals
(Section~\ref{sec:explain}). 
Our method exploits zero decomposition
techniques~\cite{wang1993elimination} to prove that
polynomial constrains over finite fields yield finite basis explanation
functions (Theorem~\ref{thm:exp:finite}), implying 
computabilty of such functions. 
We use single polynomial projections and adjust  subresultant regular subchains~\cite{wang2001elimination} to
calculate greatest common divisors with regard to partial
  variable assignments (Section~\ref{sec:mcSat:exp}), allowing
us to avoid the use of field polynomials when deriving explanation clauses during solving polynomial
constraints (Theorems~\ref{thm:projCoeff}--\ref{thm:srs}).
Our explanation clauses are integrated within \gls{mcsat},
restricting the search space of SMT solving over $\FpX$.  
We implement our approach in a new prototype for SMT solving over
finite fields (Section~\ref{sec:impl})  and experimentally demonstrate the applicability of  SMT solving over finite fields (Section~\ref{sec:experiments}).

%% file: 02_prelim.tex
\section{Preliminaries}
\label{sec:prelim}
We provide a brief summary of the relevant algebraic concepts of finite fields~\cite{contemporaryAbstractAlgebra}. %

\paragraph{Fields and Polynomials.}  
 A \emph{field} $\set F$ consists of a set $S$ on which two binary operators addition~``$+$'' and multiplication ``$\cdot$'' are defined. 
 Both operators are commutative, associative, have a neutral element in $S$ (denoted as \emph{zero} $ (0) $ and \emph{one} $ (1) $, respectively), and each element in $S$ has additive and multiplicative inverses.
Furthermore, distributivity holds. 
Informally speaking, a \emph{field} is a set $S$ with well-defined operations addition, subtraction, multiplication, and division (with the exception of division by zero).
Field examples include $\set Q$ and $\set R$. 

Let $X$ be the set of variables $\{x_1,\dots,x_n\}$.
We sort the variables in $X$ according to their index $ x_1 < x_2 < \cdots < x_n $. 
Since $x_i$ is the i-th variable in the order, we say it is of of \emph{class} $ i $, denoted by $ \cls(x_i)  = i $.
We have $X_k = \{x_i \in X \mid i \leq k\}$. 

By $\set F[X]$ we denote the ring of polynomials in variables $X$ with coefficients in $\set F$. 
A \emph{term} $\tau=x_1^{d_1}\cdots x_n^{d_n}$  is a product of powers of
variables for $d_i\in\set N$. If all $d_i = 0$, we have $\tau = 1$.
A multiple of a term $c\tau$ with $c \in \set F\setminus\{0\}$ is a \emph{monomial}. A \emph{polynomial} is a finite sum of monomials with pairwise distinct terms.

The \emph{degree of a term} $\tau$ is the sum of its exponents $\sum_{i=0}^n d_n$. The \emph{degree of a polynomial}~$p$ is the highest degree of its terms.
 We write $\ldeg{p}{x_i}$ to denote the highest \emph{degree of $x_i$} in $p$.

 For a polynomial $ p $, the set of variables of $p$ is denoted by $\var{p}$. If $\var{p} = \emptyset$, then~$p$ is \emph{constant}. If $|\var{p}| = 1$, $p$ is \emph{univariate}, and otherwise it is \emph{multivariate}.
For a set of polynomials $P$, we define $ \var{P} = \bigcup_{p\in P} \var{p} $.

 An order $\leq$ is fixed on the set of terms such that for all
 terms $\tau,\sigma_1,\sigma_2$ it holds that $1\leq\tau$ and further $\sigma_1\leq\sigma_2\Rightarrow\tau\sigma_1\leq\tau\sigma_2$.
  One such order is the \textit{lexicographic term order}:   
  If $x_1 < x_2 < \dots x_n$, then for two terms
  $\sigma_1= x_1^{d_1}\cdots x_n^{d_n}$, $\sigma_2= x_1^{e_1}\cdots x_n^{e_n}$
  it holds $\sigma_1 < \sigma_2$ iff there exists an index $i$ with $d_j = e_j$
  for all $j>i$, and $d_i<e_i$.

For a polynomial $p$, the \emph{leading variable} $ \lv{p} $ is the variable $x_i$ of $ \var{p} $ with the highest class. Let $\cls(p) = \cls(\lv{p})$.
We define the coefficient of $x_i^{\ldeg{p}{x_i}}$ as the \emph{leading coefficient} of $p$ with respect to $x_i$ and write it as $\lc{p}{x_i}$.
We denote $\red(p,x_i) = p - \lc{p}{x_i}x_i^{\ldeg{p}{x_i}}$ as the \emph{reductum} of $p$ with respect to $x_i$.

\begin{example}
  Given the polynomial $p = 2x_3^2x_1 + 4x_3x_2^4 + x_3x_2 + 7x_1 \in \set Q[x_1, x_2, x_3]$,
  we have $\var{p} = \{x_1, x_2, x_3\}$, $\lv{p} = x_3$ and $\red(p,x_3) = 4x_3x_2^2 + x_3x_2 + 7x_1$. 
  Furthermore, $\lc{p}{x_1} = 2x_3^2 + 7$, $\lc{p}{x_2} = 4x_3$, and $\ldeg{p}{x_3} = 2$, $\ldeg{p}{x_2} = 4$.
  \end{example}

  A polynomial $p\in\set F[X]$ is \emph{irreducible} if it cannot be represented as the product of two non-constant polynomials, i.e.\ there exist no $q,r \in\set F[X]$ such that $p = q\cdot r$. 
  A polynomial $g \in \set F[X]$ is called a \emph{\gls{gcd}} of polynomials $p_1,\ldots, p_s$ if $g$ divides $p_1,\ldots, p_s$
  and every common divisor of $p_1,\ldots, p_s$ divides $g$.
  
A tuple of values $ \alpha \in \set F^n $ is a \emph{root} or \emph{zero} of a polynomial $ p \in \set F[X] $ if $ p(\alpha) = 0 $.
A field $\set F$ is \emph{algebraically closed} if every non-constant univariate polynomial in $\set F[x]$ has a root in $\set F$.

Let $\set K \subseteq \set F$  be a field with respect to the field operations inherited from $\set F$. We call $\set F$ a \emph{field extension} of $\set K$ and write $\set F/\set K$.
An \emph{algebraic extension} of $\set F$ is a field extension $\set G /\set F$ such that every element of $\set G$ is a root of a non-zero polynomial with coefficients in $\set F$. 
An \emph{algebraic closure} of a field $\set F$ is an algebraic extension $\set G$ that is algebraically closed; we call $\set F$ the \emph{base field}.

\paragraph{Finite Fields.} 
In a \emph{finite field} $\Fp$ the set $S$ has only finitely many elements. The number of elements is denoted by $q$ and is called the \emph{order} of the finite field. We denote the algebraic closure of $\Fp$ as $\FpO$.
A finite field $\Fp$ exists iff $q$ is the $k$-th power of a prime $p$, i.e.\ $q = p^k$.
All finite fields with the same order are isomorphic, i.e.\ there exists a structure-preserving mapping between them. %
In case $k = 1$, $\Fp$ can be represented by the integers modulo $p$ and we have $S =\{0,1,\ldots, p-1\}$ with the standard integer addition and multiplication operation performed modulo $p$.
For example for $\set F_5$ we have $S = \{0,1,2,3,4\}$ and $2+3 = 0$ and $3\cdot4 = 2$.

The elements of $\Fp = \set F_{p^k}$ with $k>1$ are polynomials with degree $k-1$ and coefficients in~$\set F_p$.
Addition and multiplication of the polynomials is performed modulo a univariate irreducible polynomial $g \in \Fp[a]$ with degree $k$. 
For example  $\mathbb{F}_{4} =  \mathbb{F}_{2^2}$ is generated using the %
irreducible polynomial $g$: $ a^2 + a + 1 $. The elements are $\{0, 1, a, 1+a\} $ and $ ((a) + (1)) \cdot (a+1) $ evaluates to $ a $.

\paragraph{Polynomial Constraints and Formulas.}
A \emph{polynomial constraint $ f $ over $p$} in the ring $\FpX$ is of the form $ p \vartriangleright 0 $ where $ p \in \FpX$ and $ \vartriangleright\, \in \{=, \neq\} $. Since a total ordering with respect to the field operations on elements of a finite field $\Fp$ does not exist, we only consider inequality constraints of the form $ p \neq 0 $ and do not consider $<$ and $>$.
We define $ \poly{f} = p $, and extend  $ \var{f} = \var{\poly{f}}$ and $\cls(f) = \cls(\poly{f})$.
For a set of constraints $F$ we define $ \var{F} = \bigcup_{f\in F} \var{f} $. 
A polynomial constraint $ f $ is negated by substituting $ \vartriangleright $ in $f$ with the other element, i.e.\ $ \lnot (\poly{f} = 0) $ is equivalent to $ \poly{f} \neq 0 $. 
 
Let $ \nu : X \rightarrow \Fp $ denote an (partial) \emph{assignment} of variables $X$. We extend $\nu$ to an evaluation of polynomials in the natural way, i.e. $\nu(p): \FpX \rightarrow \Fp$.
Given an assignment $ \nu $ and a polynomial constraint $ f = p \vartriangleright 0 $, we say $\nu$ \emph{satisfies} $ f $ iff $ \nu(p) \vartriangleright 0 $ holds.
The function $\nu$ is also used to evaluate a constraint~$ f $.  If $ \nu $ does not assign all variables of $ \poly{f} $, we define $ \nu(f) = \sfUndef $.
If $\nu$ assigns all variables of $ \poly{f} $, then $ \nu(f) = \true $ if $ \nu $ satisfies $ f $, and $ \nu(f) = \false $ otherwise. 
Given a set of polynomial constraints $ F $, we have $\nu(F)=\sfTrue$  iff $ \nu $ satisfies all elements in $ F $. 
If such an $ \nu $ exists, we say that $ F$ is \emph{satisfiable} and $ \nu $ \emph{satisfies} $F$.

We refer to a single constraint as an \emph{atom}. A \emph{literal} is an atom or its negated form.
A~\emph{clause}~$C$ is a disjunction of literals. If $C$ contains only one literal it is a \emph{unit clause}.
A \emph{formula}~$\mathcal{F}$ is a set of clauses $\C$. Logically a formula represents a conjunction of disjunctions of literals.
An \emph{assignment} $ \nu $ satisfies a clause $C$ if at least one literal in $C$ is satisfied by $\nu$. 
Finally,  $\nu$ satisfies a set of clauses if every clause is satisfied by $\nu$.
 

%% file: 03_mcsat.tex
\section{Model Constructing Satisfiability (\glsentryname{mcsat})}
\label{sec:mcsat}
In this section, we summarize the \gls{mcsat}
approach~\cite{mcsat,mcsat_inear_int,nlsat} as presented
in~\cite{nlsat}. %
Our  \gls{mcsat} adjustments for finite fields are given in Sections~\ref{sec:explain} and~\ref{sec:mcSat:exp}.

\paragraph{\gls{mcsat} Terminology.} The \gls{mcsat} procedure is a
transition system with each state denoted by an indexed pair $ \langle
M, \C \rangle_k $ of a \emph{trail} $M$ and a set of clauses $\C$. The
index $k$ specifies the \emph{level} of the state. In our case, a
clause $C\in\C$ is a set of polynomial constraints over $\set
F_q[X]$. %
We require the following terminology: %

\renewcommand\labelitemi{--}
\begin{itemize}[wide=0em,leftmargin=1em]
	\item Each \emph{trail element} of $M$ is either a \emph{decided} or \emph{propagated literal}, or a \emph{variable assignment}. 
	\item 
	A decided literal $f$ is considered to be true. 
	A propagated literal, indicated $C \rightarrow f$, denotes that the status of clause $C$ implies that $f$ is true. 
	A variable assignment $ x_i \mapsto \alpha $ maps a theory variable $x_i \in X$ to some $ \alpha \in \Fp $.
\item We say $f \in M $, if a constraint $f$ is a trail element.
Let $\sfConstraints{M} = \{f \in M\}$.  
\item A trail is \emph{non-redundant} if it contains each constraint at most once.
\item For a constraint $f$ let $\level{f} = i \Leftrightarrow x_i \in \var{f} \land \forall j>i:  x_j \notin \var{f}$ and define the level of a clause $\level{C} = \max_{f \in C}\level{f}$. 
Let $\C_i = \{C \in \C \mid \level{C} \leq i\}$.
\item We have $\level{M} = k$, if $x_{k-1} \mapsto \alpha$ is the highest variable assignment in $M$, i.e.\ no variable assignment for $x_k, \dots, x_n$ exists.
\item A trail is \emph{increasing in level}, if all variables but the highest level variable of a constraint $f$ are assigned before $f$ appears on the trail.
\item Taking all theory variable assignments $ x_1\mapsto\alpha_1,\dots,x_k\mapsto \alpha_k $ of a trail $ M $ with $\level{M} = k+1$, we define $\alphaM = \alpha_1,\dots,\alpha_k$ and generate a (partial) assignment function $\nu_M : X \mapsto \set F_q$. We overload $\nu_M$ to evaluate constraints and sets of constraints as discussed in Section~\ref{sec:prelim}. 
\item We further say that $ M $ is \emph{feasible} if $ \nu_M(\sfConstraints{M}) $ has a solution for $ x_k $.
The set of possible values for $ x_k $ is denoted by $ \sfFeasible{M} $.
\item Given an additional constraint $f$, with $ \poly{f} \in \FpX $, we extend $ \sfFeasible{f, M} = \sfFeasible{\llbracket M, f \rrbracket}$.
If $ \sfFeasible{f, M} \neq \emptyset $ we say that $f$ is \emph{compatible} with $ M $, denoted by $ \sfCompatible{f}{M} $.
\item A state $ \langle M, \C \rangle_k $ is \emph{well-formed} when $ M $ is non-redundant, increasing in level, 
$\sfLevel{M} = k $, 
$\sfFeasible{M} \neq \emptyset$, 
$\nu_M$ satisfies $\mathcal{C}_{k-1}$, 
$\forall f \in \sfConstraints{M}: \nu_M(f) = \true $, and all propagated literals $ E\justifiy f$ are implied, i.e.\ $ f \in E $ and for all literals $ f' \neq f $ in $ E $, $ \nu_M(f') = \false $ or $ \lnot f'\in \sfConstraints{M}$.
\item
Given a well-formed state with trail $ M $, assume constraint $f$ with $ \poly{f} \in \FpXk $. Let:
\[
\sfValue{f}{M} =
\begin{cases}
\nu_{M}(f) & x_k \notin \var{f}, \level{M} = k \\
\true & f \in \sfConstraints{M} \\
\false & \lnot f \in \sfConstraints{M}\\ 
\sfUndef & \text{otherwise}
\end{cases}
\]

We overload this function to handle clauses. As such, we define
$\sfValue{C}{M} = \true$ if there exists $f \in C$ such that
$\sfValue{f}{M} = \true$;
$\sfValue{C}{M} = \false$ if 
$\sfValue{f}{M} = \false$ for al $f \in C$;
and $ \sfValue{C}{M} = \sfUndef $ in all other cases.
\end{itemize}

\begin{figure}[tb]
  \textbf{Search Rules}\vspace{-1.5em}
  \begin{center}
   \resizebox{\textwidth}{!}{\fbox{\begin{tabular}{l@{\quad}l@{$\;\rightarrow\;$}l@{\quad}l}
      [\SelectClauseRule] & $\langle M, \mathcal{C}\rangle_k$ & $\langle M, \mathcal{C}\rangle_k \vDash C $ &
        \hbox{\strut $ C \in \mathcal{C}_k \wedge \sfValue{C}{M} = \sfUndef $}\\[0.5em]

      [\ConflictRule]  & $\langle M, \mathcal{C}\rangle_k$ & $\langle M, \mathcal{C}\rangle_k \vdash C $ &
	\hbox{\strut $ C \in \mathcal{C}_k \land \sfValue{C}{M} = \sfFalse $}\\[0.5em]

      [\SatRule] & $\langle M, \mathcal{C}\rangle_k$ & $\langle\nu_M,\sfSat\rangle $ &
        \hbox{\strut $ x_k \notin \var{\mathcal{C}} $}\\[0.5em]

      [\LiftLevelRule]    & $\langle M, \mathcal{C}\rangle_k$ & $\langle\llbracket M, x_k \mapsto \alpha\rrbracket,\mathcal{C}\rangle_{k+1} $ & 
				\hbox{\strut $ x_k \in \var{\mathcal{C}} \land \alpha \in \sfFeasible{M} \land \sfValue{\mathcal{C}_k}{M} = \true $ }
	
   \end{tabular}}}
  \end{center}

  \textbf{Clause Satisfaction Rules}\vspace{-0.5em}
	\begin{center}
		\resizebox{\textwidth}{!}{\fbox{
		\begin{tabular}{l@{\quad}l@{$\;\rightarrow\;$}l@{\quad}l}
		[\BPropRule] & $\langle M, \mathcal{C}\rangle_k \vDash C$ &
	$\langle\llbracket M, C\justifiy f\rrbracket,\mathcal{C}\rangle_{k}$ &
	\lower6pt\vbox{
		\hbox{\strut $ C = \{f_1, \dots, f_m, f \} \land \sfValue{f}{M} = \sfUndef \land$ }
		\hbox{\strut $ \sfCompatible{f}{M} \land \forall i: \sfValue{f_i}{M} = \false $ }}\\[1em]
	
	[\DecideLitRule] & $\langle M, \mathcal{C}\rangle_k \vDash C$ &
	$\langle\llbracket M, f_1\rrbracket,\mathcal{C}\rangle_{k} $ &
	\lower6pt\vbox{
		\hbox{\strut $\{f_1, f_2, \dots\} \subseteq C \land \sfCompatible{f_1}{M} \land$ }
		\hbox{\strut $\forall f_i: \sfValue{f_i}{M} = \sfUndef $ }}\\[1em]
	
	[\TPropRule] & $\langle M, \mathcal{C}\rangle_k \vDash C $&
	$\langle\llbracket M, E\justifiy f\rrbracket,\mathcal{C}\rangle_{k} $ &
			\lower6pt\vbox{
				\hbox{\strut $ f \in \{L, \lnot L \mid L \in C\} \land \lnot\sfCompatible{\lnot f}{M} \land$ }
				\hbox{\strut $ \sfValue{f}{M} = \sfUndef  \land E = \sfExplain{f}{M} $ }}
			\end{tabular}}}
		\end{center}
		
  \textbf{Conflict Resolution Rules}\vspace{-1.5em}
	\begin{center}
		\resizebox{\textwidth}{!}{\fbox{
		\begin{tabular}{l@{\quad}l@{$\;\rightarrow\;$}l@{\quad}l}
	[\ResolvePropRule] & $\langle\llbracket M, E\justifiy f\rrbracket,\mathcal{C}\rangle_{k} \vdash C$ &
	$\langle M, \mathcal{C}\rangle_k \vdash R $ &
	\lower6pt\vbox{
		\hbox{\strut $\lnot f \in C \land $ }
		\hbox{\strut $R = \sfResolve{C}{E}{f} $ }}\\[0.8em]
	
		[\ResolveDecRule] & $\langle\llbracket M, f\rrbracket,\mathcal{C}\rangle_{k} \vdash C$ &
		$\langle M, \mathcal{C} \cup \{C\}\rangle_k \vDash C $ &
			\hbox{\strut $\lnot f \in C$ }\\[0.2em]
		
		[\ConsumePropRule] & $\langle\llbracket M, E\justifiy f\rrbracket,\mathcal{C}\rangle_{k} \vdash C $ &
			$\langle M, \mathcal{C}\rangle_k \vdash C$ &
				\hbox{\strut $\lnot f \in C$ }\\[0.2em]
	
		[\ConsumeDecRule] & $\langle\llbracket M, f\rrbracket,\mathcal{C}\rangle_{k} \vdash C$ &
				$\langle M, \mathcal{C} \rangle_k \vdash C$ &
					\hbox{\strut $\lnot f \in C$ }\\[0.2em]
	
	[\DropLevelOneRule] & $\langle\llbracket M, x_{k+1} \mapsto \alpha\rrbracket,\mathcal{C}\rangle_{k+1} \vdash C$ &
		$\langle M, \mathcal{C} \rangle_k \vdash C$ &
		\hbox{\strut $\sfValue{C}{M} = \false$ }\\[0.2em]

	[\DropLevelTwoRule] & $\langle\llbracket M, x_{k+1} \mapsto \alpha\rrbracket,\mathcal{C}\rangle_{k+1} \vdash C$ &
		$\langle M, \mathcal{C} \cup \{C\} \rangle_k \vDash C$ &
		\hbox{\strut $\sfValue{C}{M} = \sfUndef$ }\\[0.2em]
	
		[\UnsatRule] & $\langle\llbracket\,\rrbracket,\mathcal{C}\rangle_1 \vdash C$ &
		$\sfUnsat$ & \\
		\end{tabular}}}
	\end{center}
	\vspace{-2ex}
	\caption{Transition Rules of \gls{mcsat}}
	\label{fig:rules}
	\end{figure}

\paragraph{\gls{mcsat} Calculus.} 
The \gls{mcsat} calculus is given in Figure~\ref{fig:rules} and detailed next.
Given a set of clauses $\C$, in our case clauses of polynomial constraints over $\FpX$, the goal is to move
from an initial state of $ \langle \llbracket \, \rrbracket, \C
\rangle_1 $ to one of the two termination states, namely $
\langle\sfSat, \nu\rangle $ or $ \sfUnsat $, by continuously applying transition rules.
A termination and correctness proof that is independent of the used theory, is given in~\cite[Thm.~1]{nlsat}.

The search rules either select a clause for further processing (\SelectClauseRule), detect a conflict (\ConflictRule), detect satisfiability (\SatRule), or assign a variable while increasing the level before performing another search step (\LiftLevelRule). 
 
Clause satisfaction rules determine how a clause $ C $ is absorbed into the trail $ M $ given a state $ \langle M, \mathcal{C}\rangle_k \vDash C  $ through semantic reasoning on the theory. 
The first two rules are similar to classical DPLL-style propagation and differ in whether we meet a single compatible literal (\BPropRule) or can choose between multiple yet undetermined compatible literals (\DecideLitRule).

The \TPropRule rule is the core component of any \gls{mcsat} procedure.
It utilizes theory knowledge to propagate literals during the search.
The explanation function $\mathsf{exp}$ generates a valid lemma $E$ that justifies the propagation.
This rule was dubbed ``R-Propagation'' in~\cite{nlsat} since the
focus was merely real arithmetic. 
However, the rule itself does not rely on reals, only the explanation
function does. For our purpose, we refer to this rule as
``Theory-Propagation'', in short \TPropRule. In Section~\ref{sec:explain} we
prove that explain functions for polynomials over finite fields always
exists. Moreover, in Section~\ref{sec:mcSat:exp} we show that explanation functions are also
computable using zero decomposition procedures, avoiding the
applications of field polynomials within \gls{mcsat}. 

The conflict resolution rules of Figure~\ref{fig:rules} rely on standard boolean conflict analysis~\cite{grasp}, using the standard boolean resolution function \textsf{resolve}.
We either resolve propagation or decision steps (\ResolvePropRule, 
\ResolveDecRule) or backtrack if there is no conflicting literal in the 
trail's top literal (\ConsumePropRule, \ConsumeDecRule). 
The only theory-specific aspects of the conflict resolution are the rules \DropLevelOneRule and \DropLevelTwoRule, where we undo theory variable assignments.
We add the conflict clause $C$ to the clause set to avoid assignment repetition.

%% file: 04_explain.tex
\section{Theory Propagation for Polynomials over Finite Fields}
\label{sec:explain}

The key challenge in designing an \gls{mcsat}-based decision
procedure for a particular theory is developing theory propagation
in the respective theory to be used within the \TPropRule rule of the \gls{mcsat} calculus of
Figure~\ref{fig:rules}. In this section, we introduce zero
decomposition procedures over polynomial constraints
(Section~\ref{sec:zeroFF}) in support of theory
propagation over finite fields, allowing us to prove the existence of 
explanation clauses within \gls{mcsat} over finite fields (Section~\ref{sec:explainFF}). 

Upon application of the \TPropRule rule, a literal is selected for
propagation and justified by a newly generated explanation clause
$E$. In our work we focus on propagating polynomial constraint
literals $f$ and define their respective explanations using so-called theory
lemmas. 
\begin{definition}[Polynomial Explanation]\label{def:expclause}
	Let $f$ be a constraint, $M$ a trail, and $ E $ a clause of constraints.
	$E$ is a \emph{valid (theory) lemma} if for any arbitrary assignment $\nu$, $\nu(E) \neq \sfFalse$. 
	The clause $E$ \emph{justifies $f$ in $M$} iff $f \in E$ and $\forall f' \in E: f \neq f' \Rightarrow \sfValue{f'}{M} = \sfFalse$.
	$E$ is an \emph{explanation clause} for $ f $ in $ M $ if $E$ is a valid theory lemma and justifies~$f$~in~$M$. %
\end{definition}

Note that the explanation clauses $E$ for $f$ are generated using an
explanation function \textsf{exp} during the applications of
the \TPropRule rule of Figure~\ref{fig:rules}. We define
the \textsf{exp} function as follows. 

\begin{definition}[Polynomial Explanation Function \textsf{exp}]
	A function $\mathsf{exp}:\{$constraint$\}\times\{$trail$\}\rightarrow \{$clause$\}$ is an \emph{explanation function} $\sfExplain{f}{M} = E$
	iff $f \notin M$, $ \lnot\sfCompatible{\lnot f}{M} $, and $E$ is an explanation clause for $f$ in $M$. 
\end{definition}

\begin{example}
	A most trivial explanation function propagates $ f $ by excluding the current trail $ M $ of level $ k $ via 
	$ E = \{f\} \cup \{ \lnot f' \mid f' \in M \text{ and } x_{k-1} \in \var{f'} \} \cup \{ x \neq \alpha \mid (x\mapsto\alpha) \in M \} $.
\end{example}

As any (non-trivial) explanation function may introduce new literals,
the termination of a general \gls{mcsat} procedure requires that all
newly introduced literals are taken from a finite basis. 
\begin{definition}[Finite-basis Polynomial Explanation]
	The function $ \mathsf{exp}(f,M) $ is a \emph{finite basis explanation function} if it returns an explanation clause $ E $ for $ f $ in $ M $ and all new literals in $ E $ are taken from a finite basis.
\end{definition}

In conclusion, if a  theory admits a finite basis explanation
function, then \gls{mcsat}-based reasoning in that respective theory
is terminating. Yet, providing a finite basis explanation function is not trivial. 
A key piece %
is an efficient procedure to decompose polynomial sets.
We present our tailored procedures in Section~\ref{sec:mcSat:exp}.

\renewcommand{\S}{\ensuremath{\mathcal{S}}}
\newcommand{\T}{\ensuremath{\mathcal{T}}} 
\subsection{Zeros in Polynomials over Finite Fields}\label{sec:zeroFF}

Let us arbitrarily fix  the  sets of polynomials $
P, Q \subset \FpXk $. We assume
that $P$ contains polynomials from equality constraints, and $Q$ consists of inequality constraints.

We define the following sets of
solutions: $ \zero(P) = \{\alphaB\in\FpO^k\mid p(\alphaB) = 0 \textrm{ for all }p\in P\} $ and $ \zeroq(P) = \{\alphaB\in\Fp^k\mid p(\alphaB) = 0 \textrm{ for all }p\in P\} $.
Clearly, $\zeroq(P) \subseteq \zero(P)$.
For simplicity,  we use set subtraction to define $\zero(P / Q) = \zero(P) \setminus \zero(Q)$ and $\zeroq(P / Q) = \zeroq(P) \setminus \zeroq(Q)$.
We further write $ \hat{P} $ for $ P\setminus\FpXkm $. 
We use the tuple $\S = (P,Q)$ to mean a \emph{(polynomial) system} and write $\zero(\S) = \zero(P / Q)$.
We finally define the projection set  $\proj\zeroq(P/Q) = \{
\alphaB\in\Fp^{k-1} \mid \exists \beta\in\Fp \text{ such that }
(\alphaB,\beta) \in  \zeroq({P/Q}) \}$. Intuitively, projection sets
are used to reduce the problem of solving polynomials over $k$
variables into the smaller problem
of solving polynomials over $k-1$
variables, providing thus means for eliminating the variable $x_k$.

\begin{definition}[Zero Decomposition]\label{def:proj}
A \emph{zero decomposition procedure} is an algorithm that given $ P, Q \subset \FpXk $ generates a set of systems $\Delta = \{ (P_1, Q_1),\dots,(P_m, Q_m) \}$ such that
\begin{equation}\label{eq:zd}
\zeroq(P / Q) = \bigcup_{(P',Q')\in\Delta} \zeroq(P'/ Q').
\end{equation}
The zero decomposition procedure is \emph{projecting} in case $P',Q'\in\FpXkm$ for all $ (P',Q')\in\Delta $~and 
\[\proj\zeroq(P / Q) = \bigcup_{(P',Q')\in\Delta} \zeroq(P'/ Q').\]
Given additionally $\alphaB\in\Fp^{k-1}$ which cannot be extended to a zero, i.e.\ there is no $\beta\in\Fp$ such that $(\alphaB,\beta) \in \zeroq(P/Q)$, we say that an algorithm is a \emph{weak projecting zero decomposition procedure for $\alphaB$} if 
\begin{equation}\label{eq:wpzd}
\proj\zeroq(P / Q) \subseteq \bigcup_{(P',Q')\in\Delta} \zeroq(P'/ Q')
\end{equation}
with $ P',Q'\in\FpXkm $ and $\alphaB \notin \zeroq(P' / Q')$ for all $(P',Q')\in\Delta$.
\end{definition}
\noindent

For many zero decomposition procedures~\cite{wang2001elimination},  the
\emph{pseudo division} operation plays an important role, as follows. 
Consider polynomials $f,g \in \FpXk$, with $f \neq 0$. %
Let $r, o \in \FpXk$ denote polynomials.
We define the \emph{pseudo-remainder formula} (in $x_k$) as
\[ l^d\cdot g = o\cdot f + r \]
where $l = \lc{f}{x_k}$, $d = \textrm{max}(\ldeg{g}{x_k}-\ldeg{f}{x_k}+1,0)$, and $\ldeg{r}{x_k} < \ldeg{f}{x_k}$.
The \emph{pseudo-remainder} $r$ and \emph{pseudo-quotient} $o$ of $g$ with respect to $f$ in $x_k$ are denoted as $\prem(g, f, x_k)$ and $\pquo(g, f, x_k)$, respectively.
The polynomials $o$ and $r$ are uniquely determined by $f$ and $g$ and
are computable~\cite{wang2001elimination}.

\begin{example}
	Let $f = x_2 + x_1$ and $g = 3x_2x_1^2 + x_1$ in $\set F_5[x_1,x_2]$. Noting that $-2 = 3$ in $\set F_5$, we have eliminated $x_2$ in both the pseudo remainder and pseudo quotient by
	\[\underbrace{(3x_2x_1^2 + x_1)}_{g} \,\,\,= \underbrace{(-2x_1^2)}_{\pquo(g,f,x_2)}\cdot\,\,\, \underbrace{(x_2+x_1)}_{f} \,\,+\,\, \underbrace{(2x_1^3+x_1)}_{\prem(g,f,x_2)}.\]
\end{example}

Calculating \glspl{gcd} is another method for reducing the degree of $x_k$ and thereby eliminating it.
We employ \glspl{srs} to calculate \glspl{gcd} with respect to a partial assignment, as shown in Lemma~2.4.2 of~\cite{wang2001elimination}.
Given two polynomials $f, g \in \FpXk$ with $\ldeg{f}{x_k}\geq
\ldeg{g}{x_k} > 0$, we denote by  $ \srs(f, g, x_k) = h_2,\dots,h_r $ the \gls{srs} of $f$ and $g$ with regard to $x_k$.
Let $l = \lc{g}{x_k}$ and $l_\ell = \lc{h_\ell}{x_k} $. Then for $2\leq\ell\leq r$, we have 
\[\gcd(f(\alphaB, x_k), g(\alphaB, x_k)) = h_\ell(\alphaB, x_k)\] 
if $\alphaB \in \zero(\{l_{\ell+1},\dots,l_r\} / \{l, l_\ell\})$.
When $f$, $g$, and $h_\ell$ are partially evaluated w.r.t. $\alphaB$, the
above \gls{gcd}  is equivalent to computing a univariate \gls{gcd}.

\begin{example}\label{ex:srs}
	Let $f = x_3^2+x_3x_2+4$ and $g = x_3x_2+x_1$ in $\set F_5[x_1,x_2,x_3]$.
	Then $\srs(f,g,x_3) = [h_2, h_3] = [x_3x_2 + x_1, -x_2^2x_1 - x_2^2 + x_1^2]$.
	Using the assignment function $\nu = \{x_2\mapsto 1, x_1\mapsto 3\}$ we have 
	$\nu(f) = x_3^2 + x_3 - 1$ and $\nu(g) = \nu(h_2) = x_3 - 2$ which is indeed the \gls{gcd} of $\nu(f) $ and $ \nu(g)$.
\end{example}

\subsection{Explaining Propagated Literals}\label{sec:explainFF}
We now show that polynomial constraints over
finite fields have finite basis explanation functions.

Our \gls{mcsat}-based theory propagation works as follows. 
Let $M$ be a level~$k$ trail, implying that variable $x_k$ is not yet assigned. Suppose $f$ is a constraint such that $f \notin \sfConstraints{M}$ and $\lnot\sfCompatible{\lnot f}{M}$.
We derive polynomial constraints that are $\false$ for $\llbracket M, \lnot f \rrbracket$ to generate an explanation clause $E$ in order to justify propagating $f$. 
To enable an instant application of \TPropRule, we ensure that for all $e\in E$ either $e=f$, $\sfLevel{e} < k$, or
$\lnot e\in\sfConstraints{M}$ must hold.

\newcommand{\Ap}{\ensuremath{A_{=}}}
\newcommand{\An}{\ensuremath{A_{\neq}}}
\newcommand{\A}{\ensuremath{\mathcal{A}}}

\paragraph{Finite Basis Explanations.} Towards generating
explanation clauses $E$, we consider the polynomial constraint systems
\begin{equation}\label{eqA}
  A = \{f'\in\sfConstraints{M} \mid \sfLevel{f'} = k\} \cup \{\lnot f\}~\text{and}~
A_\vartriangleright = \{p \mid (p \vartriangleright 0) \in A\} \text{~for~}
\vartriangleright \in \{=, \neq\}.
\end{equation}
Further, we fix the 
system $\A = (\Ap,\An)$. 
From $\lnot\sfCompatible{\lnot f}{M}$ follows that
$(\alphaM,\beta)\notin\zeroq(\A)$ for all $\beta\in\Fp$.
 Based on Definition~\ref{def:proj}, we use a weak projecting
 zero decomposition procedure for $\alphaM$ and decompose $\A$ into multiple systems $\A_1,\dots,\A_r$ such that for every $1\leq\ell\leq r$ we have that $\alphaM\notin\zero(\A_\ell)$.
Then each $\A_\ell = (P_\ell,Q_\ell)$ contains (at least) one polynomial $u$ in $\FpXkm$ that excludes $\alphaM$.
Depending on whether $ u \in P_\ell $ or $u \in Q_\ell$, we generate an appropriate constraint $c_\ell$ as $u=0$ or $u\neq 0$, respectively, to ensure that $c_\ell(\alphaM) = \sfFalse$.
As a result, we set  $C = \{c_1,\dots,c_r\}$.

For any $ (\alphaB,\beta)\in\zeroq(\A)$, by Definition~\ref{def:proj} we have that $ \alphaB \in\zeroq(\A_\ell) $ for some $ 1\leq \ell\leq r $ and thus $ c_\ell(\alphaB) = \sfTrue $.
Hence, anytime an assignment function fulfills all constraints from $A$, it also fulfills at least one constraint of $C$.
We generate the explanation clause $E = \{\lnot a \mid a\in A\} \cup C$.

\begin{theorem}[Explanation Clause $E$]\label{thm:exp} Given a trail $M$ of level $k$. Let $f$ be a constraint such that $f \notin \sfConstraints{M}$ and $\lnot\sfCompatible{\lnot f}{M}$.
	Further let $A = \{f'\in\sfConstraints{M} \mid \sfLevel{f'} = k\} \cup \{\lnot f\}$ and $C = \{c_1,\dots,c_r\}$ constructed as defined above. Then 
	$E = \{\lnot a \mid a\in A\} \cup C$ is an explanation clause for $f$ in $M$.
\end{theorem} 
\begin{proof}
	By Definition~\ref{def:expclause} we show that $E$ is a valid theory lemma and justifies $f$~in~$M$.
	
	By construction we have that $\lnot f\in A$ and thus $f\in E$.
	Let $a\in A$ be a constraint such that $a \neq \lnot f$. Then $\lnot a \in \sfConstraints{M}$, therefore, $\sfValue{a}{M} = \sfFalse$.
	Let $c\in C$, from the construction of $C$ it immediately follows that $\sfLevel{c} < k$ and $c(\alphaM) = \sfFalse$, thus $\nu_{M}(c) = \sfFalse$.
	Since all constraints in $E$ but $f$ evaluate to $\sfFalse$ under $M$, we derive that \emph{$E$ justifies $f$ in $M$}.
	
	Let $\nu$ be an arbitrary assignment.
	We distinguish the following two cases:
	\begin{enumerate}[label={\it Case \arabic*:},wide=0em,leftmargin=0em]
		\item Assume $\nu(\lnot a) = \sfFalse$ for all $a \in A$.
		Let $\A = (\Ap,\An)$ as defined in (\ref{eqA}) and $\alphaB$ be $\nu$ represented as a $k$-tuple.
		Since $\nu(a) = \sfTrue$ for all $a\in A$, we have  $\alphaB\in\zero(\A)$.
		Since $\A$ was zero decomposed into systems $\A_1,\dots,\A_r$, there exists $1\leq i \leq r$ such that $\alphaB\in\zero(\A_i)$.
		Thus $c_i(\alpha) = \sfTrue$ for $c_i\in C$. As $C\subseteq E$, it follows that $\nu(E) = \sfTrue$.
		\item Assume $\nu(\lnot a) \neq \sfFalse$ for some
                  $a\in A$. As $\lnot a \in E$,  we obtain $\nu(E) \neq \sfFalse$.
                \end{enumerate}
                
	As $\nu(E) \neq \sfFalse$ in both of the cases above, we
        conclude that $E$ is a valid lemma.
\end{proof}

\begin{example}\label{ex:explain}
(Example~\ref{ex:motivating}) We have $\set F_5[x_1,x_2]$ and two %
	unit clauses $C_1 = \{c_1\} = \{x_1^2-1=0\} $ and $C_2 = \{c_2\} = \{ x_1x_2-x_2-1=0\}$.
	Assume the current trail is $ M = \llbracket x_1^2 - 1 = 0,\, x_1\mapsto 1\rrbracket $.
	We cannot add $c_2$ as we have $\lnot \sfCompatible{c_2}{M}$.
	Towards a conflict, we propagate $\lnot c_2$.
	Then $A = \{x_1x_2-x_2-1 = 0\}$, $\Ap = \{x_1x_2-x_2-1\}$, and $\An = \emptyset$.
	Using a weak zero decomposition procedure (cf. Example~\ref{ex:sdecomp}),
	we derive the zero decomposition $\Delta = \{ (\emptyset, \{x_1 - 1\}) \}$ and generate $E = \{\lnot c_2, x_1 - 1 \neq 0\}$ to justify $\lnot c_2$ on $M$.
	However, $\lnot c_2$ on $M$ results in a conflict with~$C_2$.
	Thus, we resolve $E$ with $C_2$ and learn that $x_1 - 1 \neq 0$ must hold.
	We backtrack the assignment of $x_1$ and end up with $ M = \llbracket x_1^2 - 1 = 0, x_1 - 1 \neq 0 \rrbracket $.
	Assigning $x_1 \mapsto 4$, we eventually reach \SatRule with $M = \llbracket c_1, x_1 - 1 \neq 0, x_1 \mapsto 4, c_2, x_2 \mapsto 2 \rrbracket $ and $\nu_M = \{x_1\mapsto 4, x_2\mapsto 2\}$.
\end{example}

We next show that our explanations from Theorem~\ref{thm:exp} can be
turned into explanations with finite basis. As such, 
using our explanation functions  in
the \TPropRule rule ensures that 
\gls{mcsat} terminates for our theory (Section~\ref{sec:mcSat:exp}).

\begin{theorem}[Finite Based Explanations]\label{thm:exp:finite}
	Every explanation function for theory of $ \FpX $ can be finite based.
\end{theorem}
\noindent

\begin{proof}
  The proof relies on application of Fermat's little theorem.
	We show that every polynomial $p$ in the explanation clause can be translated to an equivalent polynomial $p'$ from a finite basis.
	Given $\Fp$, by the generalized Fermat's little theorem, every element $a\in\Fp$ satisfies $a^q \equiv a$.
	Let $ t = c\prod_{i=1}^{r}x_i^{p_i} $ be a term of $p$.
	Then by Fermat's little theorem an equivalent term $t'$ can be found such that $p_i\leq q$ for all $1\leq i\leq r$.
	When $ c\in\Fp $, $r$ is finite, and $d_i \leq q$ for $1\leq i\leq r$, there is only a finite set $T$ of different terms.
	As a polynomial is a sum of terms, there are $2^{|T|}-1$ different polynomials that can be constructed from $T$.
	By replacing all terms of $p$ by an equivalent term from $T$, we have an equivalent polynomial $p'$ from a finite basis.
\end{proof}

\section{Explanation Functions over Finite Fields in \glsentryname{mcsat}}\label{sec:mcSat:exp}
Section~\ref{sec:explain} established the generation of explanation clauses  for the theory of polynomials over finite fields (Theorem~\ref{thm:exp}) and proved the existence of a finite basis explanation function (Theorem~\ref{thm:exp:finite}).

The primary component of the provided explanation generation procedure is a weak projecting zero decomposition procedure. 
It remains to show in this section that such a procedure exists. 
We provide a novel method for giving such a decomposition that does not rely on field polynomials, as in~\cite{hader-smt22}.
This is the last piece to providing a finite basis explanation function and thus to turning the \gls{mcsat} calculus of Figure~\ref{fig:rules} into an SMT solving approach over finite fields.

\paragraph{Projecting Zero Decomposition.}
Recall the generation of finite basis explanations in Section~\ref{sec:explainFF} for a trail $M$ of level $k$. Again, we have $\alphaM \in \Fp^{k-1}$ and $\A = (\Ap,\An)$ be the polynomial system of constraint set $A$ as defined in (\ref{eqA}), such that $(\alphaM,\beta)\notin\zero(\A)$ for all $\beta\in\Fp$.
Depending on $|A|$, $|\Ap|$, and $|\An|$, we utilize different
projecting procedures to find explanation clauses $E$. 
Each procedure takes $\mathcal{A}$ and $\alphaM$ as input and
decomposes $\mathcal{A}$ in the set of systems $\Delta =
\{\A_1,\dots,\A_r\}$, according to Definition~\ref{def:proj}. 
By the construction of $E$, it thus suffices to return one
constraint $f_\ell$ of each system $\A_\ell\in\Delta$ such that
$f_\ell(\alphaM)=\sfFalse$.

Based on the structure of $A$, we use single
polynomial projections (Section~\ref{sec:exp:singlePoly})
or \gls{srs}-based projections (Section~\ref{sec:exp:srs}) to derive
the explanation constraints $f_\ell$ of each system~$\A_\ell$.

\newcommand{\projCoeff}{\ensuremath{\mathsf{Proj_{Coeff}}}}
\newcommand{\projReg}{\ensuremath{\mathsf{P_{Reg}}}}
\newcommand{\projTri}{\ensuremath{\mathsf{P_{Tri}}}}

\subsection{Single Polynomial Projection for Deriving Explanation Constraints}\label{sec:exp:singlePoly}
In case $|A| = 1$ the coefficients of the polynomial constraint $f \in A$ can be used for projecting. 
By the construction of $A$ we have that $A = \{\lnot f\}$, $\level{\lnot f} = k$ and $\alphaM$ cannot be extended to satisfy $\lnot f$.
We write $\poly{\lnot f} = c_1 \cdot x_k^{d_1} + \dots + c_m \cdot x_k^{d_m}$.
By the definition of this polynomial,  we have that each $c_i\in\FpXkm$
and thus $c_i$ can be fully evaluated by $\alphaM$.
Let $\gamma_i = c_i(\alphaM)$ for $1\leq i \leq m$ and set $F = \{c_i - \gamma_i \neq 0 \mid 1\leq i \leq m\}$.
Each $f_\ell\in F$ represents one (single-polynomial) system which is returned by a zero decomposition procedure.
We denote this procedure as $\projCoeff$ and prove the following.
\begin{theorem}[Single Polynomial Weak Projection]\label{thm:projCoeff}
	Let $\A$ be a polynomial system with a single polynomial $a \in\FpXk$ and let $\alphaB\in\FpXkm$ be an assignment that cannot be extended to be a zero of $\A$.
	Then $\projCoeff(\A,\alphaB)$ is a weak projecting zero decomposition procedure for $\alphaB$.
\end{theorem}
\newcommand{\F}{\ensuremath{\mathcal{F}}}
\begin{proof}
	Termination of $\projCoeff$ is obvious.
	Let $a = c_1 \cdot x_k^{d_1} + \dots + c_m \cdot x_k^{d_m}$.
	Furthermore, there is no $\beta\in\Fp$ such that $(\alphaB,\beta) \in\zeroq(\A)$.
	Then $\projCoeff(\A, \alphaB)$ returns a set of systems $\Delta = \{(\emptyset, \{c_i - \gamma_i\}) \mid 1\leq i \leq m\}$ where $\gamma_i\in\Fp$ is $c_i(\alphaB)$.
	Let $\xiB = (\xi_1,\dots,\xi_k)\in\zeroq(\A)$. %
	Towards a contradiction, assume that $\xiB \notin \zeroq(\S)$ for all $\S\in\Delta$.
	Then for all $c_i$, we have that $c_i(\xiB) = \gamma_i = c_i(\alphaB)$, i.e.\  all coefficients of $a$ evaluate to the same value for $\alphaB$ and $\xiB$.
	As there is no value to be assigned to $x_k$ such that  $\alphaB$ can be extended to a zero of $\A$, $\xiB$ cannot exist, as $\xi_k$ would extend $\alphaB$.
	Therefore, $\xiB\in\zeroq(\S)$ for some $\S\in\Delta$.
	Furthermore, note that $\alphaB$ is excluded from all systems in $\Delta$ by construction.
\end{proof}

\begin{example}\label{ex:sdecomp}
	Filling the gap in Example~\ref{ex:explain}, we use $\projCoeff$ to decompose $\A = (\{x_1x_2-x_2-1\}, \emptyset)$ with $\alphaB = (1)$.
	Let $a$ be the polynomial in $\A$. We write $a = (x_1-1)x_2^1 + (-1)x_2^0$.
	Evaluating the coefficients, we get $\gamma_1 = 0$, $\gamma_0 = -1$ and generate 
	$F = \{(x_1-1)-0 \neq 0, (-1)-(-1) \neq 0\}$.
	As there are no zeros in the second system, we return $\Delta = \{(\emptyset,\{x_1-1\})\}$.
\end{example}

\subsection{\glsentryname{srs}-Based Projection for Deriving Explanation Constraints}\label{sec:exp:srs}
For $|A| > 1$, we use the procedure $\projReg(\A,\alphaB)$ as
shown in
Algorithm~\ref{alg:reg} and described next.
Algorithm~\ref{alg:reg} is a weak projecting zero decomposition procedure for $\alphaB$ that decomposes the system $\A$.
It utilizes \gls{srs} chains to calculate \glspl{gcd} that reduce the degree of $x_k$. This idea is based on the algorithm \textsf{RegSer} presented in~\cite{wang2001elimination,wang2000computing}.
While this original work presents \textsf{RegSer} for polynomials over
fields with characteristic $0$ only, the work
of~\cite{li2010decomposing} claims validity of  the approach also over
finite fields. Algorithm~\ref{alg:reg} relies on this result and
proceeds as follows. 

Consider two polynomials $p_1,p_2\in\FpXk$ with $\lv{p_{1}} = \lv{p_2}
= x_k$ and  let $h_2,\dots,h_r = \srs(p_1,p_2,x_k)$. Further, let $l_i = \lc{h_i}{x_k}$ for all $2\leq i \leq r$.
Then, by case distinction over the evaluation of $l_2,\dots,l_r \in \FpXkm$, the set of zeros can be decomposed to guarantee that each $h_i$ is a \gls{gcd} of $p_1$ and $p_2$ in one newly generated system. The \gls{gcd} property is then used to reduce $\ldeg{p_{1}}{x_k}$ and $\ldeg{p_{2}}{x_k}$ in this system.
The original approach of~\cite{wang2001elimination,wang2000computing} splits the zero set for every $h_i$ and thus generates exponentially many systems.
In our setting, a full decomposition can be avoided by guiding the search using $\alphaB$. 
This is done by evaluating $l_2,\dots,l_r$ with $\alphaB$ and not further exploring systems that already exclude $\alphaB$.
Therefore, only a linear amount of systems is generated in Algorithm~\ref{alg:reg}.
This computation is performed until a polynomial is found that excludes $ \alphaB $. In case there are polynomials in $x_k$ left, we exclude them using $\projCoeff$.

The \textbf{return}~\textit{call} \textbf{if}~\textit{c} statements of
Algorithm~\ref{alg:reg}, where \textit{call} is a recursive call and the \emph{guard} \textit{c} is a polynomial constraint, are used to track which path the search takes.
If $c$ evaluates to \sfTrue\xspace under $\alphaB$ then the recursive call is performed and its result is returned.
In addition, the procedure keeps a set of tracked constraints $C$ that is empty in beginning. Whenever a guard $c$ is reached but \sfFalse\xspace under $\alphaB$, it is added to $C$, otherwise, $\lnot c$ is added.
The constraints in $C$ are added to the returned sets accordingly.
Thus, the constraints in $C$ describe the search space that was not visited during the search.

\newcommand{\ReturnC}[2]{\KwRet #2 \textbf{if} #1}
\begin{algorithm}[!tb]
	\small
	\DontPrintSemicolon
	\caption{\projReg($\A = (P,Q)$, $\alphaB$)}
	\label{alg:reg}
	\Return $(\{p\}, \emptyset)$ for any $p\in P$ with $\lv{p} < x_k$ and $p(\alphaB) \neq 0$\label{alg:reg:ret1}\;
	\Return $(\emptyset, \{q\})$ for any $q\in Q$ with $\lv{q} < x_k$ and $q(\alphaB) = 0$\label{alg:reg:ret2}\;

		\If{$|\hat{P}| > 0$}{
			select $p\in P$ with the smallest positive $\ldeg{p}{x_k}$\label{alg:reg:a}\;

			\ReturnC{${\lc{p}{x_k}(\alphaB) = 0}$}{\projReg($P\setminus\{p\}\cup \{\red(p,x_k)\}, Q, \alphaB$)}\label{alg:reg:5}\;

			\If{$|\hat{P}| > 1$}{\label{alg:reg:ps}
				select any $p'\in \hat{P}\setminus\{p\}$\;
				compute $h_2,\dots,h_r = \srs(p',p,x_k)$ and let $l_i = \lc{h_i}{x_k}$ for $2\leq i\leq r$\;

				\If{$\lv{h_r} < x_k$}{
				  \ReturnC{$l_{r-1}(\alphaB) \neq 0$}{\projReg($P\setminus\{p,p'\}\cup\{h_r, h_{r-1}\}, Q, \alphaB$)}\label{alg:reg:special1}\;
				  $r\gets r - 2$\;
				}

				\lFor{$i=r,\dots,2$}{\label{alg:reg:pe}%
				  \ReturnC{$l_i(\alphaB) \neq 0$}
				  {\projReg($P\setminus\{p,p'\}\cup\{h_i,l_{i+1},\dots,l_r\}, Q, \alphaB$)}
				}

			}
			\ElseIf{$|\hat{Q}| > 0$ and $\lv{p} = x_k$}{
				select any $q\in \hat{Q}$\label{alg:reg:qs}\;
				compute $h_2,\dots,h_r = \srs(q,p,x_k)$ and let $l_i = \lc{h_i}{x_k}$ for $2\leq i\leq r$\label{alg:reg:ex}\;

				\If{$\lv{h_r} < x_k$}{
					\ReturnC{$l_r(\alphaB) \neq 0$}{\projReg($P\setminus\{p\}\cup\{\pquo(p,h_r,x_k)\}, Q\setminus\{q\}, \alphaB$)}\label{alg:reg:special2}\;
					$r\gets r - 1$\;
				}
				\lFor{$i=r,\dots,2$}{%
				  \ReturnC{$l_i(\alphaB) \neq 0$}{\projReg($P\setminus\{p\}\cup\{\pquo(p,h_i,x_k), l_{i+1},\dots,l_r\}, Q, \alphaB$)}
				}\label{alg:reg:qe}
			}

			\Return $ \projReg((P,Q) \cup \projCoeff(p, \alphaB),\alphaB) $\label{alg:reg:end1}\;
		}
		\ElseIf{$|\hat{Q}| > 0$}{
		  \lForAll{$q\in\hat{Q}$}{
		    \ReturnC{$\lc{q}{x_k}(\alphaB) = 0$}{\projReg($P, Q\setminus\{q\}\cup\{\red(q, x_k)\}, \alphaB$)}\label{alg:reg:qq}
		  }
		\Return $ \projReg((P,Q\setminus \hat{Q}) \cup \projCoeff(\prod_{q'\in \hat{Q}} q', \alphaB),\alphaB) $
		}
	
	\end{algorithm}

Recall the notion of $\hat{P}, \hat{Q}$ from Section~\ref{sec:zeroFF}.
Lines~\ref{alg:reg:ret1}--\ref{alg:reg:ret2}  of Algorithm~\ref{alg:reg} return an excluding polynomial in case one is found.
Line~\ref{alg:reg:5} ensures that $\lc{p}{x_2} \neq 0$ which is a requirement for any further \gls{gcd} operation.
Lines~\ref{alg:reg:ps}-\ref{alg:reg:pe} are used to remove polynomials of $P$ until only one is left, which is then used in lines~\ref{alg:reg:qs}-\ref{alg:reg:qe} to remove polynomials from $Q$.
Lines~\ref{alg:reg:special1} and \ref{alg:reg:special2} handle the special case  $\lv{h_r} < x_k$. By definition of \gls{srs}, $\lv{h_i} = x_k$ for $2\leq i \leq r-1$, but not necessarily for $h_r$.
Roughly speaking, $\lv{h_r} < x_k$ denotes a constant \gls{gcd} and, thus, divisor-free polynomials.
Line~\ref{alg:reg:qq} splits elements in $Q$ to remove $x_k$ in case all polynomials $p\in P$ are free of~$x_k$.

\begin{example}
	Assume we have the system $\A = (\{x_3^2+x_3x_2+4\},\{x_3x_2+x_1\})$
	in $\set F_5[x_1,x_2,x_3]$ and let $\alphaB = (3,1)$.
	At line~\ref{alg:reg:ex}, $\projReg$ we will calculate the first SRS according to Example~\ref{ex:srs}.
	Eventually, the computation terminates with a zero decomposition represented by the constraints $\{x_2 = 0, -x_2^2x_1 - x_2^2 + x_1^2\neq 0, -x_2^4 + 2x_2^2x_1 \neq 0\}$, each representing one generated system.
\end{example}

\begin{theorem}[\gls{srs}-Based Weak Projection]\label{thm:srs}
	Let $\A$ be a polynomial system and let $\alphaB\in\Fp^{k-1}$ be an assignment that cannot be extended to be a zero of $\A$.
	Then $\projReg(\A, \alphaB)$ of Algorithm~\ref{alg:reg} is a weak projecting zero decomposition of $\A$ for~$\alphaB$.
\end{theorem}

\begin{proof}
	We show that $\projReg(\A, \alphaB)$ terminates and is a weak projecting zero decomposition for $ \alphaB $.
	
	\medskip
	\noindent \textit{Termination:} As the first two loops in Algorithm~\ref{alg:reg} are bound by the size of the \gls{srs} decomposition $r$ and the size of a \gls{srs} is bound, both loops certainly terminate.
	The third and last loop iterates over the finite amount of elements in $\hat{Q}$ and thus terminates.
	It remains to show that the recursion depth of Algorithm~\ref{alg:reg} is bound.
	Note that for every recursive call of $\projReg$ the degree in $x_k$ for at least one polynomial in $\A = (P,Q)$ decreases.
	Once $ \hat{P} = \hat{Q} = \emptyset $ no further recursive call is performed.
	We distinct two cases:
	\begin{enumerate}[label={\it Case \arabic*:}, wide=0em,leftmargin=1em]
		\item\label{prf:reg:term:case1}
		Assume $|\hat{P}| > 0$. Then, the degree of $x_k$ in
		polynomials of $P$ is reduced in each recursive call of
		lines~\ref{alg:reg:a}-\ref{alg:reg:qe} of Algorithm~\ref{alg:reg}.
		In case one polynomial $p\in \hat{P}$ remains, we use $ \projCoeff $ to remove $ x_k $.
		
		\item\label{prf:reg:term:case2} 
		Assume $|\hat{P}| = 0$.
		Algorithm~\ref{alg:reg} proceeds by splitting polynomials in $\hat{Q}$ in line~\ref{alg:reg:qq}.
		For a given polynomial $q\in\hat{Q}$ the recursion depth of the call in line~\ref{alg:reg:qq} is bound by the number of coefficients of $q$ in $x_k$.
		For each call the leading coefficient is removed. Once $x_k$
		is removed, we have $\lc{q}{x_k} = q$.
		Then, we either return in line~\ref{alg:reg:ret2} or the guard of the call in line~\ref{alg:reg:qq} is \sfFalse.
	\end{enumerate}

	\noindent
	As all recursive calls eventually terminate,
	Algorithm~\ref{alg:reg} terminates.
	Note that from the zero decomposition argument below follows that the recursion always ends in line~\ref{alg:reg:ret1} or \ref{alg:reg:ret2}. This usually happens before all $ x_k $ are eliminated.
	
	\medskip 
	\noindent \textit{Zero Decomposition:}
	Results of~\cite{wang2001elimination,wang2000computing} imply
	that \textsf{RegSer} is a zero decomposition procedure that generates a sequence of regular systems.
	Besides other properties, for a regular system $(P,Q)$ holds that either $\hat{P} = \emptyset$ or $\hat{Q} = \emptyset$. Furthermore, $P$ is a triangular set, thus $|\hat{P}| \leq 1$.
	With a very similar argument can be proven that $\projReg$ performs a zero decomposition towards regular systems, although the systems are not fully computed.
	In $\projReg$ the decomposition ends after $x_k$ has been fully processed as for generating explanations further decomposition is not required.
	
	Let $\A' = (P,Q)$ be one decomposed system from $\A$.
	We first show that $ P,Q\in\FpXkm $ and $ \alphaB \notin\zeroq(\A')$.
	From the lemmas presented in~\cite{wang2001elimination} for
	\textsf{RegSer}, it follows that the each decomposition step in lines~\ref{alg:reg:a}-\ref{alg:reg:qe} of Algorithm~\ref{alg:reg} as well as line~\ref{alg:reg:qq} performs a zero decomposition according to
	equation~(\ref{eq:zd}); thus, $\alphaB$ cannot be extended to a zero of any such generated system.
	In case $ \A' $ is a systems that was not further expanded in a conditional recursive call, i.e.\ the negation of the guard is in $\A'$, then the desired property holds by construction.
	In case $ \A' $ contains a polynomial from $ \FpXkm $ which
	excludes $ \alphaB $ directly, Algorithm~\ref{alg:reg} stops
	in lines~\ref{alg:reg:ret1} or \ref{alg:reg:ret2}, returning
	only this one polynomial.
	In case the regular decomposition procedure has concluded for $ x_k $ and no such polynomials can be found, by definition of a regular system,
	we end up with either exactly one polynomial $p\in\hat{P}$ or $\hat{Q} \neq \emptyset$, but not both.
	We distinct two cases:
	\begin{enumerate}[label=(\alph*), wide=0em,leftmargin=1em]
		\item
		Assume $\hat{Q} = \emptyset$, then $\hat{P} = \{p\}$.
		Since $ \alphaB $ cannot be extended to a zero of $ \A' $ but is not excluded by any other polynomial in $ \A' $, we conclude that it cannot be extended to become a zero of $ p $.
		Therefore, we may call $\projCoeff$ in line~\ref{alg:reg:end1} to further decompose $ \A' $.
		The weak projecting zero decomposition property of $ \projCoeff $
		concludes the proof. 
		
		\item
		Assume $\hat{Q} \neq \emptyset$, then $\hat{P} = \emptyset$.
		Since the regular decomposition process has concluded, the recursive call in line~\ref{alg:reg:qq} is not executed for any $ q\in\hat{Q} $.
		We know that $ \alphaB $ cannot be extended such that all $ q\in\hat{Q} $ evaluate to a non-zero value, we have that the product of all $ q\in\hat{Q} $ when evaluated with $ (\alphaB,\beta) $ for all $ \beta\in\Fp $.
		We thus use $ \projCoeff $ to concludee the weak projecting
		zero decomposition property. 
	\end{enumerate}
	
	\smallskip
	
	We finally show that $ \projReg $ fulfills equation~\ref{eq:wpzd}.
	Let $ \xiB \in \zeroq(\A) $.
	As $\projReg$ performs a zero decomposition, there is a system $ \A' $ such that $ \xiB\in\zeroq(\A') $.
	If Algorithm~\ref{alg:reg} returns in
	lines~\ref{alg:reg:ret1} or \ref{alg:reg:ret2},  then $ \xiB $ is a zero of the returned single polynomial (sub-)system of $ \A' $.
	If $ \A' $ is not further expanded because of a guard $ c $ in an conditional recursive call, the polynomial of $ \lnot c $ is in the according set of $ A' $ and returned as a single polynomial sub-system of $ \A' $. As $ \xiB $ is a zero of $ \A' $, it is also a zero of the returned sub-system.
	Finally, in case $ \projReg $ utilizes $ \projCoeff $ to remove a polynomial in $ x_k $ from $ \A' $, we have by Theorem~\ref{thm:projCoeff} that $ \xiB $ is a zero of one of the returned (projected) systems.
	In any case, $ \xiB $ is a zero of the decomposition and thus equation~\ref{eq:wpzd} holds.
\end{proof}

%% file: 05_implementation.tex
\section{Implementation}\label{sec:impl}
We have implemented our \gls{mcsat} approach for SMT solving over finite fields in a new prototype\footnote{The source code of the prototype together with the generated test instances are available: \href{https://github.com/Ovascos/ffsat}{https://github.com/Ovascos/ffsat}},
written in Python and using the computer algebra system Sage~\cite{sage} for handling polynomials.
While our work is not limited to a specific field order, practical implementation constraints (from our implementation as well as Sage) are a limiting factor in the prototype's ability to handle large(r) field. Besides the general procedure of \gls{mcsat} and the theory specific details presented in Sections~\ref{sec:explain} and~\ref{sec:mcSat:exp}, the performance of our approach therefore depends on certain implementation details. In the sequel we discuss our design decisions.

\paragraph{Selecting literals for propagation.}
While the general \gls{mcsat} framework does not restrict the application of theory propagation beyond the conditions of the \TPropRule rule, it is up to the theory to determine whether a theory propagation is applicable and appropriate. 
For the theory of polynomials over finite fields, we utilize a similar propagation strategy as \cite{nlsat} uses for reals.
Let $ \langle M, \C\rangle_{k} $ be the current state with feasible trail $ M $ and $ f \in C $ a literal of a previously selected (yet unsatisfied) clause $ C \in \C $ such that $ \sfLevel{f} = k $ and $ \sfValue{f}{M} = \sfUndef $.
If this happens, we use \TPropRule to add $ f $ to $ M $ in order to satisfy $ C $.

Let $ \mathcal{X}_k \subseteq \Fp $ be the set of possible values for $ x_k $ that satisfy $ \nu[M](f) $.
We can distinguish four different scenarios of propagation by comparing feasible values for $ x_k $, namely:

\begin{enumerate}[label=(\roman*)]
	\item If $ \mathcal{X}_k = \Fp $, then $ f $ is propagated.
	\item If $ \mathcal{X}_k = \emptyset $, then $ \lnot f $ is propagated.
	\item If $ \mathcal{X}_k \supseteq \sfFeasible{M} $, i.e.\ $ f $ does not restrict the feasible values of $ M $, then  $ f $ is propagated. %
	\item If $ \mathcal{X}_k \cap \sfFeasible{M} = \emptyset $, then $ \lnot f $ is propagated. %
\end{enumerate}

Informally, a propagation of $f$ demonstrates the (theory) knowledge that $C$ is fulfilled by $M$. 
By propagating $\lnot f$ with explanation $E$, it is very likely that a conflict will arise right away (cf. Example~\ref{ex:explain}). The design of $E$ results in an immediate resolution of $E$ with $C$. Because generating explanation clauses is costly, they are not generated at the moment of propagation but only when they are needed for conflict analysis.

\paragraph{Storing feasible values.}
As we work with  finite set $\Fp$ of theory values, feasible values for a reasonable small $q$ can be enumerated. Even for larger field orders, the number of zeros of a polynomial given a partial assignment is still constrained by the length of the polynomial.

\paragraph{Variable Order.}
The order in which the theory variables are assigned in the trail can hugely influence the number of conflicts and thus generated explanations.
Finding a beneficial variable order is a general consideration for both \gls{mcsat} style approaches and computer algebra algorithms alike.
While it is highly important for practical performance, our procedure is correct for any ordering; optimizing the variable order is an interesting task for future work.

%% file: 06_experiments.tex
\newcommand{\gb}{Gr\"obner basis\xspace}
\newcommand{\tsR}{\textsc{Rand}\xspace}
\newcommand{\tsC}{\textsc{Craft}\xspace}
\newcommand{\pFF}{\textsc{FFSat}\xspace}
\newcommand{\pGB}{\textsc{GB}\xspace}

\section{Experiments and Discussion}\label{sec:experiments}
We compare the performance of our Python prototype in solving polynomial systems to state-of-the-art \gb techniques provided by Sage.
It is important to note that by design \gb techniques require
polynomial systems (cf. Section~\ref{sec:zeroFF}) as inputs and are
incapable of handling polynomial constraints (i.e.\ disjunctions and
conjunctions of polynomials). 
We therefore limit our experimental comparison to polynomial
benchmarks with a small subset of inputs, as 
\gb algorithms with general polynomial constraints  over finite fields would
involve exponential many calls (in the number of constraints) to the \gb algorithm. %

We further note that SMT-LIB standard and repository~\cite{SMTLIB}
does not support finite field arithmetic. For this reason, we
cannot yet directly compare our work to SMT-based solvers. To
circumvent such limitations, 
 we represent polynomial constraints directly in our Python framework
 and compare our work only  to \gb 
approaches supporting such input. 

\paragraph{Experimental setup.} For our experiments on SMT solving
over finite fields,
we created $250$ polynomial systems over a range of finite field orders ($3$, $13$, $211$) %
and different numbers of variables (up to $64$). 
To get better insights, we have utilized two different methods of polynomial system generation:
\begin{itemize}[wide=0em,leftmargin=1em]
	\item \tsR:
	All polynomials in this test set are fully random generated by Sage's \texttt{random\_element} function.
	The degree of the polynomials is at most $4$.
	These created systems are more frequently unsatisfiable and have fewer zeros on average. This category of tests has smaller systems and fewer variables since they are challenging to solve (for any strategy).
	It is ensured that at least one polynomial has a constant term to avoid trivial $0$-solutions, as this would give our approach an unfair advantage.
	\item \tsC:
	These polynomial systems are crafted to have multiple solutions by explicitly multiplying zeros.
	They tend to be easy to solve. %
	Thus, these systems are considerable larger with a huge amount of variables.
	Polynomial constraints are restricted to up to $5$ distinct variables with up to $3$ zeros each.
\end{itemize}

\noindent
Each test set consists of 25 polynomial systems with
fixed field order and a fixed number of variables and constraints; see
Table~\ref{tbl:experiments}. Our experiments were run on an AMD EPYC 7502 CPU 
with a timeout of $300$ seconds per benchmark instance.

We compare our procedure (\pFF) to a \gb approach (\pGB). The latter uses 
field polynomials to limit the solutions to those for the base field.
To get an elimination ideal and thus ensure to get a satisfiable
assignment from a calculated \gb, one typically relies on  lexicographic term-ordering, which is especially expensive.
However, to ``only'' check whether a polynomial system is satisfiable without returning an assignment, it suffices to calculate the \gb in any term ordering.
In our experiments we therefore calculate two different bases. \pGB uses the (efficient) default ordering provided by Sage, while $\pGB_\textsc{lex}$ uses a lexicographic term ordering.

\paragraph{Experimental results.}
For analyzing our experimental findings, let us note that the already
highly engineered \gb algorithms written in C/C++ utilized by Sage
have an inherent performance advantage compared to our Python
implementation. Yet, Table~\ref{tbl:experiments} demonstrates that our
approach works well for satisfiable cases.
 The number of instances that were resolved between \pFF and \pGB
 within the predetermined timeout of 300s for each instance is also  compared in Table~\ref{tbl:experiments}.
Each test is identified by its type, the finite field size $q$, the number of variables $n$, and the number of constraints per system $c$.

\begin{table}
	\begin{center}
		\begin{tabular}{c|ccc|ccc}
			Type & $q$ & $n$ & $c$ & \pFF & \pGB & $\pGB_\textsc{lex}$ \\
			\hline
			\tsR & $3$ & $8$ & $8$ & \textbf{25} & \textbf{25} & \textbf{25} \\
			\tsR & $3$ & $16$ & $16$ & \textbf{12} & 11 & 0 \\
			\tsC & $3$ & $32$ & $32$ & \textbf{25} & \textbf{25} & 0 \\
			\tsC & $3$ & $64$ & $64$ & \textbf{25} & 24 & 0 \\
			\tsR & $13$ & $8$ & $4$ & \textbf{25} & 0 & 0 \\
			\tsR & $13$ & $8$ & $8$ & \textbf{1} & 0 & 0 \\
			\tsC & $13$ & $32$ & $16$ & \textbf{19} & 18 & 1 \\
			\tsR & $211$ & $8$ & $4$ & \textbf{17} & 0 & 0 \\
			\tsR & $211$ & $8$ & $16$ & 0 & 0 & 0 \\
			\tsC & $211$ & $16$ & $8$ & 24 & \textbf{25} & \textbf{25} \\
		\end{tabular}
	\end{center}
	\caption{Instances solved by \pFF, \pGB, and $\pGB_\textsc{lex}$, out of
25 polynomial systems per test set.}
	\label{tbl:experiments}
\end{table}

\paragraph{Experimental analysis and discussions.}
We note the following key insights of our approach in comparison to \gb approaches:
\begin{itemize}[wide=0em,leftmargin=1em]
	\item Our Python prototype can already keep up with highly engineered \gb approaches on some classes of instances
	but further engineering work is required to match existing \gb implementations consistently.
	\item The strength of our work comes with solving satisfiable
          instances. This is because we can often find a satisfying
          assignment without fully decomposing the polynomial system.
	\item While the \gls{mcsat} approach is capable of detecting conflicts by deriving empty clauses, the point in time when an empty clause can be derived is highly dependent on the variable order.
	\item The lack of an order on finite fields leads to the
          generation of many inequality constraints when a partial
          assignment cannot be extended. Developing further
          optimizations to detect such cases, especially for unsat
          instance with large field orders, is a task for future
          work. 
	\item \gb approaches are saturation based, thus they have a
          conceptually advantage on unsatisfiable instances.
					This is due to the fact that \gb methods terminate once a non-zero constant is determined, which is why we feel inventing and employing extremely efficient monomial orderings would aid our work. 
	\item It is notable that our approach seems to show complementary performance characteristics to existing \gb techniques, indicating that a portfolio approach could be valuable.
\end{itemize}
In summary, our current experiments show the general effectiveness of
our approach, indicating also how present  weaknesses can be mitigated with existing techniques from \gls{mcsat} and \gls{cdcl} solving (e.g.\ heuristics on the variable order, restarts, pre- and improcessing techniques to reduce clause complexity, clause deletion, etc.).

%% file: 07_related.tex
\section{Related Work}

Gröbner bases~\cite{Buchberger} and triangular sets~\cite{aubry1999theories, 
triagSetOverview} have been introduced to compute  the solution space
of polynomial equations,  by reducing the degree of polynomials through
variable elimination. %
Solving  polynomial equations in general entails finding all of its
solutions in the algebraic closure of the underlining coefficient field. 
Yet, for the purpose of satifiability, solutions in the base field
are usually of the most interest.
Obviously, there are only a finite number of solutions if the base
field is finite; yet, enumerating all of the finitely numerous
possibilities is not practically viable. 

To limit the solutions of Gröbner bases and triangular sets to finite
fields, a common technique is to introduce
and add the set of field polynomials to the set of polynomial
equations~\cite{ffgb,huang2012elim}
Using field polynomials though greatly impacts practical performance,
as showcased in~\cite{hader-smt22}. 
Specialized ways for computing 
Gröbner bases and triangular sets over finite fields have therefore been created, 
such as the XL algorithm~\cite{XLAlgo}, F4~\cite{f4}, and 
F5~\cite{f5} for Gröbner bases.
Although all of these strategies are aimed at solving polynomial
systems over finite fields, none of them explicitly address
inequalities even though inequalities may be converted into equalities
using the Rabinowitsch trick~\cite[4.2 Prop.~8]{CoxLittleOShea-Book07}. %

Optimization concepts for triangular sets have been introduced in~\cite{wang2001elimination}, including efficient characteristic set 
algorithms~\cite{gao2012characteristic,huang2012elim} and polynomial decomposition 
into simple sets~\cite{li2010decomposing}.
Although these approaches integrate reasoning over inequalities, none
of them considers systems of clauses with polynomial constraints as
needed for our SMT solving problem. 
Furthermore, they all require the generation of exponentially many sets to fully describe the systems. Our approach only explores a linear sized decomposition on demand.

A related approach to our search procedure is given in the hybrid
framework of~\cite{bettale2012solving,bettale2009hybrid}.
Here, a partial evaluation of the system is performed by fixing some
variables before starting multiple Gröbner bases
computations. Instead, in our work we show that subresultant regular
subchain computations allows us to avoid working with  Gröbner bases
(and hence their double-exponential computational complexities).

Substantial progress has also been devoted to the problem 
of dealing with specific boolean polynomials, i.e.\ finite fields with only two 
elements.
PolyBoRi~\cite{brickenstein2009polybori,bettale2012solving} is fairly 
effective in this domain, but it does not generalize towards arbitrary
finite fields, which is the focus of our work. 

Recently, an algebraic SMT decision technique for computing
satisfiability of polynomial equalities/inequalities over large
prime fields has been introduced in~\cite
{vellaalt-smt22}. As polynomial systems are a subset of our polynomial
constraint clauses, our work complements this effort, by also
establishing a computational approach for deriving explanation clauses
within \gls{mcsat} reasoning. 

%% file: 08_conclusion.tex
\section{Conclusion}

We introduce a novel  reasoning approach for determining the satisfiability of a given system of non-linear polynomial constraints over finite fields. 
As a framework, we adopt an \gls{mcsat} decision procedure and expand it with a specific theory propagation rule that allows variable propagation over finite fields by adding so-called explanation clauses.
To show the existence of these explanation clauses over finite fields, we apply zero decomposition procedures over polynomial constraints.
Based on the structure of the polynomial system, we construct explanation clauses to resolve conflicting variable assignments.
We distinguish between 
single polynomial projections and projections of multiple polynomials using subresultant regular subchains. 
Our work avoids using field polynomials while reducing the size of the projected polynomials.

We aim to further optimize our prototype through specific design decisions. For example, we will investigate the effect the variable order has on SMT solving over finite fields. 
Furthermore, we wish to improve performance if the given polynomial system is unsatisfiable; in this case, we are also interested in generating proof certificates. Finally, integrating our prototype within a high-performance \gls{smt} solver is another line for future work. 

%% file: ffsat.bbl
\begin{thebibliography}{10}
\providecommand{\url}[1]{\texttt{#1}}
\providecommand{\urlprefix}{URL }
\providecommand{\doi}[1]{https://doi.org/#1}

\bibitem{Reduce}
{Anthony C. Hearn and the REDUCE developers}: Reduce,
  \url{http://www.reduce-algebra.com/}, accessed: 03-13-2023

\bibitem{aubry1999theories}
Aubry, P., Lazard, D., Maza, M.M.: {On the Theories of Triangular Sets}.
  Journal of Symbolic Computation  \textbf{28}(1-2),  105--124 (1999)

\bibitem{triagSetOverview}
Aubry, P., Maza, M.M.: {Triangular Sets for Solving Polynomial Systems: a
  Comparative Implementation of Four Methods}. Journal of Symbolic Computation
  \textbf{28}(1-2),  125--154 (1999)

\bibitem{SMTLIB}
Barrett, C., Fontaine, P., Tinelli, C.: {The Satisfiability Modulo Theories
  Library (SMT-LIB)}. {\tt www.SMT-LIB.org} (2016)

\bibitem{bettale2009hybrid}
Bettale, L., Faugere, J.C., Perret, L.: {Hybrid Approach for Solving
  Multivariate Systems over Finite Fields}. Journal of Mathematical Cryptology
  \textbf{3}(3),  177--197 (2009)

\bibitem{bettale2012solving}
Bettale, L., Faug{\`e}re, J.C., Perret, L.: {Solving Polynomial Systems over
  Finite Fields: Improved Analysis of the Hybrid Approach}. In: ISSAC. pp.
  67--74 (2012)

\bibitem{brickenstein2009polybori}
Brickenstein, M., Dreyer, A.: {PolyBoRi: A framework for Gr{\"o}bner-basis
  Computations with Boolean Polynomials}. Journal of Symbolic Computation
  \textbf{44}(9),  1326--1345 (2009)

\bibitem{Buchberger}
Buchberger, B.: {Bruno {Buchberger's PhD} Thesis 1965: An Algorithm for Finding
  the Basis Elements of the Residue Class Ring of a Zero Dimensional Polynomial
  Ideal}. Journal of Symbolic Computation  \textbf{41}(3-4),  475--511 (2006)

\bibitem{XLAlgo}
Courtois, N., Klimov, A., Patarin, J., Shamir, A.: {Efficient Algorithms for
  Solving Overdefined Systems of Multivariate Polynomial Equations}. In:
  EUROCRYPT. pp. 392--407. Springer (2000)

\bibitem{CoxLittleOShea-Book07}
Cox, D., Little, J., O'Shea, D.: Ideals, Varieties, and Algorithms.
  Springer-Verlag New York (1997)

\bibitem{DBLP:journals/jsc/DavenportH88}
Davenport, J.H., Heintz, J.: {Real Quantifier Elimination is Doubly
  Exponential}. Journal of Symbolic Computation  \textbf{5}(1/2),  29--35
  (1988)

\bibitem{mcsat}
De~Moura, L., Jovanovi{\'c}, D.: {A Model-Constructing Satisfiability
  Calculus}. In: VMCAI. pp. 1--12. Springer (2013)

\bibitem{f4}
Faugere, J.C.: {A New Efficient Algorithm for Computing Gr{\"o}bner Bases
  (F4)}. Journal of Pure and Applied Algebra  \textbf{139}(1-3),  61--88 (1999)

\bibitem{f5}
Faugere, J.C.: {A New Efficient Algorithm for Computing Gr{\"o}bner Bases
  without Reduction to Zero (F5)}. In: ISSAC. pp. 75--83 (2002)

\bibitem{contemporaryAbstractAlgebra}
Gallian, J.A.: {Contemporary Abstract Algebra}. Chapman and Hall/CRC (2021)

\bibitem{ffgb}
Gao, S., Platzer, A., Clarke, E.M.: {Quantifier Elimination over Finite Fields
  using Gr{\"o}bner Bases}. In: Algebraic Informatics. pp. 140--157 (2011)

\bibitem{gao2012characteristic}
Gao, X.S., Huang, Z.: {Characteristic Set Algorithms for Equation Solving in
  Finite Fields}. Journal of Symbolic Computation  \textbf{47}(6),  655--679
  (2012)

\bibitem{GoldwaserMicaliRackoff-SIAM89}
Goldwasser, S., Micali, S., Rackoff, C.: {The Knowledge Complexity of
  Interactive Proof Systems}. SIAM Journal on Computing  \textbf{18}(1),
  186--208 (1989)

\bibitem{ma}
Hader, T.: Non-Linear {SMT}-Reasoning over Finite Fields. Master's thesis, TU
  Wien, Vienna (Feb 2022)

\bibitem{hader-smt22}
Hader, T., Kov{\'{a}}cs, L.: {An {SMT} Approach for Solving Polynomials over
  Finite Fields}. In: {SMT}. {CEUR} Workshop Proceedings, vol.~3185, pp.
  90--98. CEUR-WS.org (2022)

\bibitem{ecc}
Hankerson, D., Menezes, A.J., Vanstone, S.: {Guide to Elliptic Curve
  Cryptography}. Springer-Verlag, Berlin, Heidelberg (2003)

\bibitem{huang2012elim}
Huang, Z.: {Parametric Equation Solving and Quantifier Elimination in Finite
  Fields with the Characteristic Set Method}. Journal of Systems Science and
  Complexity  \textbf{25}(4),  778--791 (2012)

\bibitem{ecsda}
Johnson, D., Menezes, A., Vanstone, S.: {The Elliptic Curve Digital Signature
  Algorithm {(ECDSA)}}. Int. Journal of Information Security  \textbf{1}(1),
  36–63 (2001)

\bibitem{mcsat_inear_int}
Jovanovi{\'c}, D., De~Moura, L.: {Cutting to the Chase Solving Linear Integer
  Arithmetic}. In: CADE. pp. 338--353 (2011)

\bibitem{nlsat}
Jovanovi{\'c}, D., De~Moura, L.: {Solving Non-Linear Arithmetic}. In: IJCAR.
  pp. 339--354 (2012)

\bibitem{li2010decomposing}
Li, X., Mou, C., Wang, D.: {Decomposing Polynomial Sets into Simple Sets over
  Finite Fields: The Zero-Dimensional Case}. Computers \& Mathematics with
  Applications  \textbf{60}(11),  2983--2997 (2010)

\bibitem{book:MW}
MacWilliams, F., Sloane, N.: The Theory of Error-Correcting Codes.
  North-holland Publishing Company, 2nd edn. (1978)

\bibitem{maple}
{Maplesoft, a division of Waterloo Maple Inc..}: Maple,
  \url{https://hadoop.apache.org}, accessed: 03-13-2023

\bibitem{Matiyasevich93}
Matiyasevich, Y.: Hilbert’s Tenth Problem. The MIT Press, Cambridge (1993)

\bibitem{tls}
Moeller, B., Bolyard, N., Gupta, V., Blake-Wilson, S., Hawk, C.: {Elliptic
  Curve Cryptography {(ECC)} Cipher Suites for Transport Layer Security
  {(TLS)}}. RFC 4492 (2006)

\bibitem{mou-phdthesis}
Mou, C.: {Solving Polynomial Systems over Finite Fields: Algorithms,
  Implementation and Applications}. Phd thesis, {Universit{\'e} Pierre et Marie
  Curie} (May 2013)

\bibitem{NiuWang-MCS}
Niu, W., Wang, D.: {Algebraic Approaches to Stability Analysis of Biological
  Systems}. Mathematics in Computer Science  \textbf{1},  507--539 (03 2008)

\bibitem{ozdemir2023satisfiability}
Ozdemir, A., Kremer, G., Tinelli, C., Barrett, C.: Satisfiability modulo finite
  fields. Cryptology ePrint Archive  (2023)

\bibitem{grasp}
Silva, J.P.M., Sakallah, K.A.: {GRASP—A New Search Algorithm for
  Satisfiability}. In: The Best of ICCAD, pp. 73--89. Springer (2003)

\bibitem{ssh}
Stebila, D., Green, J.: {Elliptic Curve Algorithm Integration in the Secure
  Shell Transport Layer}. RFC 5656 (2009)

\bibitem{sage}
Stein, W., et~al.: {S}age {M}athematics {S}oftware,
  \url{http://www.sagemath.org}, accessed: 03-13-2023

\bibitem{SturmfeldsSolvingSystemofPolyEqs}
Sturmfels, B.: Solving Systems of Polynomial Equations. No. no. 97 in CBMS
  Regional Conference Series in Mathematics, AMS (2002)

\bibitem{Szabo2018SmartC}
Szabo, N.: {Smart Contracts: Building Blocks for Digital Markets} (1996),
  [Online]. Available: http://www.fon.hum.uva.nl

\bibitem{vellaalt-smt22}
Vella, L.C., Alt, L.: {On Satisfiability of Polynomial Equations over Large
  Prime Fields}. In: {SMT}. {CEUR} Workshop Proceedings, vol.~3185, pp.
  114--127. CEUR-WS.org (2022)

\bibitem{wang1993elimination}
Wang, D.: {An Elimination Method for Polynomial Systems}. Journal of Symbolic
  Computation  \textbf{16}(2),  83--114 (1993)

\bibitem{wang2000computing}
Wang, D.: {Computing Triangular Systems and Regular Systems}. Journal of
  Symbolic Computation  \textbf{30}(2),  221--236 (2000)

\bibitem{wang2001elimination}
Wang, D.: {Elimination Methods}. Springer Science \& Business Media (2001)

\bibitem{Mathematica}
{Wolfram Research{,} Inc.}: Mathematica, {V}ersion 13.2,
  \url{https://www.wolfram.com/mathematica}, accessed: 03-13-2023

\end{thebibliography}
